\documentclass{lmcs}


\usepackage{hyperref}
\usepackage{cleveref}
\usepackage{microtype}
\usepackage{graphicx}
\usepackage[utf8]{inputenc}
\usepackage{bm}
\usepackage{amsfonts}
\usepackage{amssymb}
\usepackage{pgfplots}
\usepackage{amsmath,amsthm}
\usepackage{tikz-cd}
\usetikzlibrary{nfold}
\usepackage{mathtools}
\usepackage{xfrac}
\usepackage{bm}
\usepackage{wasysym}
\usepackage{xcolor}
\usepackage{mathrsfs}
\usepackage{xurl}
\usepackage{thmtools}
\usepackage{scrextend}
\usepackage{mdframed}
\usepackage{comment}
\usepackage{tcolorbox}
\usepackage{quiver}
\usepackage{mathdots}
\usepackage{calligra}
\usepackage{adjustbox}
\usepackage{mleftright}

\makeatletter
\renewcommand*\env@matrix[1][*\c@MaxMatrixCols c]{%
  \hskip -\arraycolsep
  \let\@ifnextchar\new@ifnextchar
  \array{#1}}
\makeatother

\pgfplotsset{
  unit circle/.style={
    width=4cm,
    height=4cm,
    axis lines=middle,
    xtick=\empty,
    ytick=\empty,
    axis equal,
    enlargelimits,
    xmax=1,
    ymax=1,
    xmin=-1,
    ymin=-1,
    domain=0:pi/2
  }
}

\newcommand{\C}{\mathcal C}
\newcommand{\D}{\mathcal D}
\newcommand{\E}{\mathcal E}
\newcommand{\U}{\mathcal{U}}
\newcommand{\V}{\mathcal{V}}
\newcommand{\F}{\mathcal F}
\newcommand{\K}{\mathcal K}
\renewcommand{\L}{\mathsf L}
\newcommand{\Z}{\mathbb Z}

\newcommand{\dom}{\textit{dom}}
\newcommand{\cod}{\textit{cod}}
\newcommand{\nat}{\mathsf{nat}}
\newcommand{\base}{\mathsf{base}}
\newcommand{\loopp}{\mathsf{loop}}
\newcommand{\J}{\mathcal{J}}
\newcommand{\refl}{\mathsf{refl}}
\newcommand{\ind}{\mathsf{ind}}
\newcommand{\inl}{\mathsf{inl}}
\newcommand{\inr}{\mathsf{inr}}
\newcommand{\0}{\mathbf{0}}
\newcommand{\1}{\mathbf{1}}
\newcommand{\2}{\mathbf{2}}
\newcommand{\lN}{\mleft\lVert}
\newcommand{\rN}{\mright\rVert}
\newcommand{\modal}{{\ensuremath{\ocircle}}}

\DeclareMathOperator{\auto}{\mathsf{auto}}
\DeclareMathOperator{\cofiber}{\mathsf{cof}}
\DeclareMathOperator{\assoc}{\mathsf{assoc}}
\DeclareMathOperator{\transport}{\mathsf{transp}}
\DeclareMathOperator{\ap}{\mathsf{ap}}
\DeclareMathOperator{\apd}{\mathsf{apd}}
\DeclareMathOperator{\pr}{\mathsf{pr}}
\DeclareMathOperator{\glue}{\mathsf{glue}}
\DeclareMathOperator{\homm}{\mathsf{hom}}
\DeclareMathOperator{\op}{op}
\DeclareMathOperator{\colimm}{\mathsf{colim}}
\DeclareMathOperator{\limm}{\mathsf{lim}}
\DeclareMathOperator{\homogpth}{\textsf{homog-pth}}
\DeclareMathOperator{\colim}{\mathrm{colim}}
\DeclareMathOperator{\idd}{\mathsf{id}}
\DeclareMathOperator{\fun}{\mathsf{fun}}
\DeclareMathOperator{\bp}{\mathsf{bp}}
\DeclareMathOperator{\ty}{\mathsf{ty}}
\DeclareMathOperator{\pt}{\mathsf{pt}}
\DeclareMathOperator{\obb}{\mathsf{Ob}}
\DeclareMathOperator{\ev}{\mathsf{ev}}
\DeclareMathOperator{\postcomp}{\mathsf{postcomp}}

\newcommand{\myhref}[2]{%
  \href{#1}{\textcolor{blue!30!black}{#2}}%
}
\theoremstyle{definition}\newtheorem{note}[thm]{Note}

\begin{document}

\title{On Left Adjoints Preserving Colimits in Homotopy Type Theory}
\titlecomment{{\lsuper*}This paper is an extended version of \cite{lapccsl}.}

\author[P.~Hart]{Perry Hart\lmcsorcid{0000-0002-7247-362X}}

\address{University of Minnesota, Minneapolis, MN, USA}
\email{hart1262@umn.edu}

\begin{abstract}
We examine how the standard proof that left adjoints preserve colimits behaves in the setting of wild categories, a natural setting for synthetic homotopy theory inside homotopy type theory. We show that the proof may fail for adjunctions between wild categories and even produce a wild left adjoint that fails to preserve colimits. Our core contribution, however, is a sufficient condition on the left adjoint for the proof to go through. The condition, which we call \emph{$2$-coherence}, expresses that the naturality structure of the hom-isomorphism commutes with composition of morphisms. We present two useful examples of this condition in action. First, we use it, along with a new version of a known trick for homogeneous types, to show that the suspension functor, as well as a generalization thereof, preserves graph-indexed colimits. Second, we show that every modality, viewed as a functor on coslices of a type universe, is $2$-coherent as a left adjoint to the forgetful functor from the subcategory of modal types, thereby proving this subcategory is cocomplete. We have formalized our main results in Agda.
\end{abstract}

\maketitle

\section{Introduction}\label{sec:intro} 

In category theory, a basic and eminently useful fact is that \emph{left adjoints preserve colimits (LAPC)}. We would like to invoke this classical theorem in the categorical setting of synthetic homotopy theory, the axiomatic and usually type-theoretic study of topological spaces with higher-dimensional structure. This would let us produce new universal constructions of spaces from existing ones via purely algebraic methods. For synthetic homotopy theory carried out in homotopy type theory (HoTT), the appropriate categorical setting is that of \emph{wild categories}, the canonical examples of which are type universes. This notion is a type-theoretic approximation of an $\mleft(\infty,1\mright)$-category that specifies the data of a $1$-category but omits higher coherence data. Despite its naivete, this setting is expressive enough to study concepts like (co)limits and adjunctions inside type theory. 

We would like to port \emph{LAPC} to adjunctions between wild categories. In particular, we would like to port the ``standard'' proof, i.e., the proof one would expect to see based on the category theory literature. This, however, turns out to be harder than one might hope: we produce, inside HoTT, an example of such an adjunction for which the proof fails. In fact, we produce something stronger: an adjunction between wild categories that simply fails to preserve colimits. Nevertheless, we identify a sufficient condition for the proof to go through. Roughly, it expresses that the adjunction data interacts nicely with the left adjoint's (proof-relevant) composition law. With this condition, combined with a higher version of a known technique based on homogeneous types, we show that suspension (and a generalization of it), a known example of a left adjoint, preserves (graph-indexed) colimits, which has applications to the theory of acyclic types and to homology theory. We also show that every modality (such as truncation), viewed as a functor on coslices of a type universe, satisfies the condition as a left adjoint to the forgetful functor. As a result, the full wild subcategory of modal types inherits colimits from the ambient coslice. Our proof of this fact differs from the one described for truncations by \cite[Section 7.4]{Uni13}. Ours is ultimately simpler by placing modalities in the general context of left adjoints. We have formalized all our main results in Agda.\footnote{See \cite{2cohagda-pub} for the Agda codebase. One may type-check the mechanized proofs for this paper by running \texttt{Dockerfile.lmcs}, which builds the appropriate version of Agda and then takes about 25 minutes to type-check the relevant files.} 

\subsection{Motivation}

To motivate our work, we should review the well-known proofs of \emph{LAPC} from classical category theory and explain why the one we choose is the right one to port to wild categories. In addition, we should explain why the issue of porting it deserves the HoTT community's attention. We assume the reader is familiar with the basic notions of category theory.

Consider an adjunction $L \dashv R$ between $1$-categories $\C$ and $\D$. Let $\J$ be a small $1$-category. The classical theorem states that $L$ preserves $\J$-shaped colimits, and it has two well-known proofs. The first requires that $\C$ and $\D$ admit global colimit functors $\colim_{\J} : \C^{\J} \to \C$ and $\colim_{\J} : \D^{\J} \to \D$~\cite[Section V.5]{CTWM}. The proof assumes these colimit functors satisfy some coherence conditions that are automatically true for $1$-categories (but \emph{not} for wild ones). It proceeds by using the uniqueness of left adjoints to define an isomorphism $\varphi : \colim_{\J} \circ L^{\J} \cong  L \circ \colim_{\J}$. By unfolding the units of two composite adjunctions derived from $L \dashv R$ and showing that $\varphi$ commutes with these units, we can deduce that $\varphi$ maps the canonical cocone on $\colim_{\J}(L(F))$ to the induced one on $L(\colim_{\J}(F))$ for all diagrams $F : \J \to \C$. Proving that $\varphi$ commutes in this way tacitly uses coherence conditions on $L \dashv R$ that hold for $1$-categories. Also, to conclude that $L(\colim_{\J}(F))$ is colimiting, the proof tacitly uses the pentagon identity of $\D$ (which holds trivially) to transfer the colimiting property of $\colim_{\J}(L(F))$.

Instead of requiring global colimit functors, the second proof starts with a specific colimit $\colim_{\J}(F)$ of a diagram $F : \J \to \C$ and shows that $L(\colim_{\J}(F))$ is a colimit under $L(F)$. Like the first proof, it secretly uses coherence conditions that hold for $1$-categories, a point we'll return to. It argues that for all $Y \in \obb(\D)$, the following chain of isomorphisms with $C \coloneqq \colim_{\J}(F)$ equals the canonical post-composition map~\cite[Theorem 4.5.2]{CTIC}:
\[ \label{eq:iso}
  \homm_{\D}(L(C), Y) \cong \homm_{\C}(C, R(Y)) \cong \limm_i(\homm_{\C}(F_i, R(Y))) \cong \limm_i(\homm_{\D}(L(F_i), Y))\tag{\texttt{iso}}
\]
This means that the induced cocone on $L(\colim_{\J}(F))$ is indeed colimiting, i.e., $L$ preserves colimits. Besides avoiding global colimit functors, this proof argues directly in terms of hom-isomorphisms and their naturality conditions---data directly supplied by the usual definition of adjunctions between wild categories. This helps us formulate further laws that such adjunctions must satisfy for the proof to work. Finally, it doesn't rely on bicategorical structure of $\C$ or $\D$. Overall, the second proof, which we call the \emph{standard proof}, is better for wild categories.

So, what makes porting the standard proof an interesting problem? The chain of isomorphisms \eqref{eq:iso} is not hard to replay in wild category theory. Due to the secret coherence conditions, however, proving that it equals the canonical map becomes a problem. In fact, this equality is sometimes false. This problem is surprising at first glance and easy to miss. It indicates a subtle mismatch between the data used to construct the cocone on $L(\colim_{\J}(F))$ and the data used to construct a hom-type adjunction. (In our setting, $\homm$ is a family of types, not necessarily sets.) The latter uses only the $0$- and $1$-dimensional data of $L$, the data coming from the underlying directed graph of $\C$. The former, however, also uses the composition law of $L$, a $2$-dimensional datum. This mismatch has strange effects: we can build two naturally isomorphic left adjoints such that the standard proof goes through for one but not the other. One of the chief virtues of this paper is that it puts this issue into the literature and sets the record straight on the approach to \emph{LAPC} that was expected to work.

\subsection{Contributions}

In this section, we explain the contributions of the paper and its organization. We start by outlining the heart of the paper: a coherence condition on the data of a (hom-type) adjunction between wild categories that guarantees the left adjoint preserves colimits. Afterward, we outline two applications of this coherence condition in synthetic homotopy theory. Finally, we explain what is new in this paper compared with \cite{lapccsl}. Throughout the paper, we'll provide hyperlinks, in blue, to associated Agda code. 

\subsubsection{$2$-coherent left adjoints (\autoref{2Adj})}

Let $\C$ and $\D$ be wild categories. (We review the relevant concepts of wild category theory in \cref{wildcats}.) Our work centers on a new notion of coherence for adjunctions between $\C$ and $\D$.  Such an adjunction consists of functors $L : \C \to \D$ and $R : \D \to \C$ together with
a family of type equivalences $\psi : \prod_{X : \obb(\D)}\prod_{ A : \obb(\C)}\homm_{\D}(L(A), X) \xrightarrow{\simeq} \homm_{\C}(A,R(X))$ and witnesses $\nat_{\cod}$ and $\nat_{\dom}$ of the naturality of $\psi$ in $X$ and in $A$. The main point of this paper is that, despite being a direct translation of the classical notion, this definition is not coherent enough, due to the proof-relevant identity structure of the hom types. Indeed, it doesn't let us prove left adjoints preserve colimits, the defining property of left adjoints between locally presentable categories.
To solve this problem, we introduce the following coherence condition.

Given an adjunction between $\C$ and $\D$, we say that $L$ is \emph{$2$-coherent} if for all suitable morphisms $h_1$, $h_2$, and $h_3$, the identity $\psi(h_1) \circ h_2 \circ h_3 = \psi(h_1 \circ L(h_2 \circ h_3))$ obtained by applying $\nat_{\dom}$ multiple times equals the identity obtained by applying the composition law $L_{\circ}$ of $L$. 
This nice interaction between $\nat_{\dom}$ and $L_{\circ}$ holds automatically for $1$-categories, in which case these two equalities are proof-irrelevant. Also, it holds for adjunctions between $\mleft(\infty,1\mright)$-categories by virtue of the infinite tower of coherence data they encode. For wild categories, however, adjunctions may fail to satisfy it.

We prove that $2$-coherent left adjoints preserve colimits. The proof proceeds entirely by algebraic manipulation of the adjunction data. Conceptually, it is quite direct as it relies on just three foundational ingredients of HoTT: the naturality of homotopies, the triangle identity of equivalences, and the structure identity principle. 

During the proof, we must take care to eliminate the data witnessing that $\psi$ is an equivalence. Otherwise, we would need to include this data in the $2$-coherence condition, which would make it far less tractable. Indeed, the condition we arrive at is relatively simple to check and thus quite useful in practice. We just need to know how to compute $\nat_{\dom}$ and $L_{\circ}$. In \cref{suspsum,modsum}, we outline natural examples of left adjoints along with methods for proving that they are $2$-coherent.

\subsubsection{Suspension preserves colimits (\autoref{SuspCol})} \label{suspsum}

The suspension endofunctor $\Sigma : \U_{\ast} \to \U_{\ast}$ on the wild category of pointed types is  critical to synthetic homotopy theory. For a while, it has been known that $\Sigma$ is left adjoint to the loop space functor $\Omega$ in HoTT. Its preservation of colimits has been expected to follow from \emph{LAPC} in the usual way. Yet, this paper shows that the usual way requires a  coherence between the adjunction $\Sigma \dashv \Omega$ and $\Sigma$'s composition law, which makes the proof trickier than expected. 

We verify that this coherence holds, i.e., that $\Sigma$ is $2$-coherent, thereby verifying that $\Sigma$ preserves colimits. (As we'll see, our proof applies, for each pointed type $X$, to the \emph{join} $X \ast {-} : \U_{\ast} \to \U_{\ast}$, which is equivalent to $\Sigma$ when $X \equiv \2$ and is critical to synthetic homotopy theory in its own right~\cite{brunthesis,rijkejoin}.) In this case, the final part of the proof of $2$-coherence is infeasible to perform  directly. Instead, we get it for free by proving a new, higher version of Cavallo's trick for homogeneous types, which include all loop spaces. 

This infeasibility highlights a major difference between our type system---Book HoTT---and cubical type theory~\cite{cubagda}. Unlike Book HoTT, cubical supports \emph{definitional} $\beta$-rules for path constructors in higher inductive types (HITs), such as suspension types. Such support greatly simplifies the proof in question as it erases many postulated equalities that we must carry around.  Although constructions with HITs tend to be much harder in Book HoTT~\cite{cubsyn}, they are still valuable. Indeed, Book HoTT has models in all $\mleft(\infty,1\mright)$-toposes~\cite{HITsem, topoimod}.\footnote{Closure of universes under parameterized HITs still needs to be fully worked out.} By contrast, it's not yet known whether Cubical Agda's type theory has a model Quillen equivalent to a model category presenting spaces. Moreover, cubical is an extension of Book HoTT~\cite[Section 2.16]{ABCHFL}, so we can interpret our results into it.

Inside HoTT, the fact that $\Sigma$ preserves colimits has useful consequences. It implies that the pointed acyclic types~\cite{acyclic} are closed under colimits in $\U_{\ast}$. Moreover, it puts on firm footing a key step of Graham's construction of stable homotopy as a homology theory: proving that $\Sigma$ preserves cofibers~\cite[Corollary 2.2]{graham}. 

\subsubsection{Colimits of modal types (\autoref{ModCol})} \label{modsum}

Modalities are functors $\modal : \U \to \U_{\modal}$ on a type universe $\U$ that arise as reflectors into well-behaved subuniverses $\U_{\modal}$ of $\U$~\cite{RS}. Although they are well-studied in HoTT~\cite{modaldes, modalchar}, their interaction with colimits could be explained better. Since reflectors are by definition left adjoints, we should expect that modalities preserve colimits.
In the case of pushouts, this property already has a proof for $n$-truncations $\lN{-}\rN_n$~\cite[Section 7.4]{Uni13}, which are examples of modalities, and it should extend easily to all graph-indexed colimits. The problem is that this proof, which we call the \emph{Book proof}, relies on the particular computational behavior of the composition law of $\lN{-}\rN_n$. Thus, it doesn't generalize to arbitrary left adjoints.

We offer a different proof of the same property by showing that every modality $\modal$ is a $2$-coherent left adjoint. By fitting into the general framework of \emph{LAPC}, our proof gets rid of the ad hoc steps required by the Book proof. As a result, ours  amounts to an easy application of the induction principle of $\modal$. Moreover, we argue that ours matches the usual classical proof for reflective subcategories because one would normally view colimit preservation here as a special case of \emph{LAPC}.\footnote{This goes against \cite[Section 7.7]{Uni13}'s claim that the Book proof matches the usual classical one.} Hence our proof has some advantages over the Book proof.

In fact, given a modality $\modal$, we prove the more general fact that the induced functor $\modal^A : A/\U \to \mleft(A/\U\mright)_{\modal}$ on an arbitrary coslice of $\U$ is $2$-coherent. (We recover the original case by setting $A$ to $\0$.) Previously, Hart and Favonia used that $\lN{-}\rN_n^A$ is $2$-coherent in order to construct colimits of higher groups~\cite[Section 7]{CSL25}. Here we offer a concise proof that is formulated for all modalities. Finally, we use the colimit preservation of $\modal^A$ together with some general results about \emph{wild bicategories} to build colimits in $\mleft(A/\U\mright)_{\modal}$ from those in $A/\U$, the latter of which are explicitly constructed by \cite[Section 5]{CSL25}.

\subsubsection*{New contributions}
The current paper expands \cite{lapccsl} with two new mathematical contributions. First, we strengthen the counterexample to the standard proof by presenting a new wild left adjoint that provably fails to preserve colimits (\autoref{truecexmp}), settling the question of whether wild left adjoints preserve colimits. We also prove that the original example does preserve colimits, thereby showing that $2$-coherence is not a necessary condition for cocontinuity. Second, we extend the $2$-coherence proof for the suspension to the (unary) join functor $X \ast {-}$ (\autoref{joinpres}). 

We have formalized both of these new contributions in Agda. Ultimately, the only result in this paper left unformalized is the original counterexample to the standard proof (\autoref{eq-counter}).

\section{Additional related work}

\subsection{Wild category theory}

Our work establishes a fundamental property of well-behaved adjunctions between wild categories. Motivated by synthetic homotopy theory, it fits neatly into the theory of wild categories developed for this purpose~\cite{Hur, RK, smashmon}. More generally, it continues the study of adding higher coherence data to wild-categorical notions. So far, this study has focused on adding conditions internalizing the axioms of a $\mleft(2,1\mright)$-category to the wild category itself~\cite{chen, CK, CSL25}. The literature calls the resulting concept \emph{wild $2$-precategory}, \emph{$2$-coherent wild category}, or \emph{wild bicategory}.
In this paper, we focus on a different aspect of $2$-coherence: adding it to data \emph{between} wild categories.   
\subsection{Homogeneous types}

For the application to the suspension functor (\autoref{SuspCol}), we build on the theory of \emph{homogeneous types} in HoTT~\cite[Section 2]{banded}. These are pointed types that are independent of basepoint in a strong sense. The key feature of such types is that proving identities about pointed maps into them is considerably easier than about arbitrary pointed maps. This feature, known as \emph{Cavallo's trick}, can make normally intractable computations in higher path algebra tractable. For example, Ljungstr\"om adapted Cavallo's trick to show that the smash product forms a symmetric monoidal product on the wild category of pointed types~\cite{smashmon}. We provide a new adaptation to handle the $2$-coherence condition for the suspension and the join.

\section{Background on type theory}\label{Backg}

We review some basic constructions in HoTT that are important for our work. We assume the reader is familiar with Martin-L\"of type theory (MLTT), the core type system of HoTT, in the style of~\cite{Uni13}. Notably, MLTT is sufficient for the core of our work: all of \autoref{2Adj} is carried out in MLTT (with function extensionality). For \cref{SuspCol,ModCol}, we postulate a simple class of HITs: pushout types~\cite[Section 6.8]{Uni13}. In particular, \autoref{SuspCol} focuses on \emph{suspensions}, a kind of pushout.

\subsection{Type system}

We review three constructions in our type system. The first is the function $\ap_f : \mleft(x = y\mright) \to \mleft(f(x) = f(y)\mright)$ defined by path induction for all functions $f : X \to Y$ and $x, y : X$. (We use $=$ for the identity/path type and $\equiv$ for definitional equality.) If we view $X$ as an $\infty$-groupoid, then $\ap$ is the action of $f$ on morphisms of $X$, thereby exhibiting $f$ as a functor. A key feature of $\ap$ is the following naturality law.
\begin{lem}[Homotopy naturality] \label{homnat}
Let $f, g : X \to Y$. For all $x, y : X$, $p : x = y$, and $H : f \sim g$, we have a commuting square of identities
\[\begin{tikzcd}
	{f(x)} && {g(x)} \\
	{f(y)} && {g(y)}
	\arrow["{H(x)}", Rightarrow, no head, from=1-1, to=1-3]
	\arrow["{\ap_f(p)}"', Rightarrow, no head, from=1-1, to=2-1]
	\arrow["{\ap_g(p)}", Rightarrow, no head, from=1-3, to=2-3]
	\arrow["{H(y)}"', Rightarrow, no head, from=2-1, to=2-3]
\end{tikzcd}\] 
\end{lem}
Here, $f_1 \sim f_2 \coloneqq \prod_{x : X} f_1(x) = f_2(x)$ for any dependent functions $f_1, f_2 : \prod_{x: X}Y(x)$, called the type of \emph{homotopies} between $f_1$ and $f_2$. If $f_1 \sim f_2$, we say that $f_1$ and $f_2$ are \emph{homotopic}.

The second is the notion of \emph{half-adjoint equivalence}. Let $f : X \to Y$ be a function. We say that $f$ is a \emph{half-adjoint equivalence}, or just \emph{equivalence}, if it has a function $g : Y \to X$, homotopies $\eta_f : g \circ f \sim \idd_X$ and $\epsilon_f : f \circ g \sim \idd_Y$, and a triangle identity $\tau_f(x) : \ap_f(\eta_f(x)) = \epsilon_f(f(x))$ for every $x : X$. We may denote $g$ by $f^{-1}$. We use $\simeq$ to refer to equivalences. A function is an equivalence if and only if it is \emph{bi-invertible}, i.e., has a right inverse $s$ and a left inverse $r$~\cite[Corollary 4.3.3]{Uni13}. Moreover, if  a function is an equivalence, then it has contractible fibers~\cite[Section 4.4]{Uni13} (i.e., each of its fibers is equivalent to $\1$).

The third is the \emph{transport} function $\transport^Y : \prod_{x, y : X}\prod_{p : x = y}Y(x) \to Y(y)$ for any type family $Y$ over $X$.  This notion gives us a dependent version of $\ap$: if $f : \prod_{x : X}Y(x)$, then we have a function $\apd_f : \prod_{x, y : X}\prod_{p : x = y} \transport^Y(p, f(x)) = f(y)$. As a result, we can generalize \autoref{homnat} as follows: for all $f, g : X \to Y$, $x, y : X$, $p : x = y$, and $q : f(x) = g(x)$, we have a path $\ap_f(p) \cdot \transport^{f \sim g}(p,q) = q \cdot \ap_g(p)$. The transport function is essential for stating the induction principle of HITs, such as suspensions.

\subsection{Suspensions} \label{suspdef}

For all functions $f : X \to Y$, the \emph{cofiber $\cofiber(f)$ of $f$} is the pushout of the span $\1 \leftarrow X \xrightarrow{f} Y$. Let $X$ be a type. The \emph{suspension $\Sigma(X)$ of $X$} is the cofiber of $X \to \1$. Explicitly, it is the pushout
\[
\begin{tikzcd}[column sep = 34, row sep = 30]
	X & \1 \\
	\1 & {\Sigma(X)}
	\arrow[from=1-1, to=1-2]
	\arrow[""{name=0, anchor=center, inner sep=0}, from=1-1, to=2-1]
	\arrow[""{name=1, anchor=center, inner sep=0}, "\inr", from=1-2, to=2-2]
	\arrow["\inl"', from=2-1, to=2-2]
	\arrow["\lrcorner"{anchor=center, pos=0.125, rotate=180}, draw=none, from=2-2, to=1-1]
	\arrow["\glue"{marking, allow upside down}, draw=none, from=0, to=1]
\end{tikzcd}
\] where $\glue : \prod_{x : X}\inl(\ast) = \inr(\ast)$. We denote $\inl(\ast)$ and $\inr(\ast)$ by $\mathsf{N}$ and $\mathsf{S}$, respectively, and consider $\mathsf{N}$ the basepoint of $\Sigma(X)$. The induction principle for $\Sigma(X)$ states that for every type family $E$ over $\Sigma(X)$ with elements
\begin{equation*}
\begin{aligned}[c]
 t_{\mathsf{N}} &  \ : \ E(\mathsf{N})
\\  t_{\mathsf{S}} &  \ : \ E(\mathsf{S})
\end{aligned}
\qquad \qquad
\begin{aligned}[c]
  T &  \ : \ \prod_{x : X}\transport^E(\glue(x), t_{\mathsf{N}}) = t_{\mathsf{S}}
\end{aligned}
\end{equation*}
we have a function $\ind(E, t_{\mathsf{N}}, t_{\mathsf{S}}, T) : \prod_{z : \Sigma(X)}E(z)$ that satisfies the definitional equalities 
\[
\ind(E, t_{\mathsf{N}}, t_{\mathsf{S}}, T)(\mathsf{N})  \ \equiv \ t_{\mathsf{N}} \qquad \ind(E, t_{\mathsf{N}}, t_{\mathsf{S}}, T)(\mathsf{S})  \  \equiv \  t_{\mathsf{S}}
\] and is equipped with a \emph{typal $\beta$-rule}: an identity $\beta_{\ind(E, t_{\mathsf{N}}, t_{\mathsf{S}}, T)}(x) : \apd_{\ind(E, t_{\mathsf{N}}, t_{\mathsf{S}}, T)}(\glue(x)) = T(x)$. In the non-dependent case, this principle is called the \emph{recursion principle}.

Let $\mleft(X,x_0\mright), \mleft(Y, y_0\mright) : \U_{\ast}$ be pointed types in a universe $\U$ and $\mleft(f, f_0\mright) : \mleft(X,x_0\mright) \to_{\ast} \mleft(Y,y_0\mright)$ be a pointed map. We have a pointed map $\Sigma(f, f_0) : \Sigma(X, x_0) \to_{\ast} \Sigma(Y,y_0)$ defined by recursion on $\Sigma(X)$, which trivially preserves the basepoint. Note that $\beta_{\Sigma(f, f_0)}(x) : \ap_{\Sigma(f, f_0)}(\glue(x)) = \glue(f(x))$ for each $x : X$.

\section{Wild category theory} \label{wildcats}

In this section, we record essential concepts and constructions in wild category theory, including adjunctions. The reader will notice that the basic definitions are naive translations of their classical counterparts.

The key distinction between wild categories and the categories of \cite[Section 9.1]{Uni13} is that the latter have hom types that behave as \emph{sets}, i.e., have trivial identity types, so that all higher coherences between morphisms hold trivially. In this way, the latter internalize $1$-categories in HoTT. By contrast, wild categories simply ignore such coherences. In synthetic homotopy theory, many wild categories have hom types that are not sets, so developing the theory of wild categories is worthwhile.

\begin{defi}
  A \emph{wild category (relative to universes $\U$ and $\V$)} is a tuple consisting of a type $\obb : \U$ of objects, a family $\homm : \obb \to \obb \to \V$ of hom types, identity morphisms $\idd$, a composition operation $\circ$, left and right unit laws for $\circ$, and an associativity law $\assoc$ for $\circ$.
  
  Here, the unit and associativity laws are paths: in particular, $\assoc(f,g,h) : \mleft(f \circ g\mright) \circ h = f \circ \mleft(g \circ h\mright)$ for all composable morphisms $f$, $g$, and $h$.
\end{defi}

\begin{exa}
Let $A$ be a type. The wild category $A/\U$ has objects $\sum_{X : \U}A \to X$ and morphisms $X \to_A Y \coloneqq \sum_{k : \pr_1(X) \to \pr_1(Y)}k \circ \pr_2(X) \sim \pr_2(Y)$. Composition and associativity are defined easily via path induction (and the characterization of identity types of $\Sigma$-types). The wild category of pointed types $\U_{\ast}$, which is isomorphic to $\1/\U$, has a similar structure.
\end{exa}

The following notion generalizes the univalence axiom~\cite[Axiom 2.10.3]{Uni13} and can significantly simplify proofs. We will invoke it in the proof of \autoref{modcolbuild}.
\begin{defi} \label{unvwc}
  A wild category $\C$ is \emph{univalent} if for all $A , B : \obb(\C)$, the canonical function $\mleft(A = B\mright) \to \mleft(A \simeq_{\C} B\mright)$ is an equivalence. Here, elements of the right-hand type are \emph{equivalences in $\C$}, defined as bi-invertible morphisms.
\end{defi}

\begin{defi}
Let $\C$ and $\D$ be wild categories.

\begin{enumerate}
\item A \emph{(wild) functor} $F : \C \to \D$ from $\C$ to $\D$ consists of a function $F_0 : \obb(\C) \to \obb(\D)$ and an action on morphisms $F_1  : \homm_{\C}(X,Y) \to \homm_{\D}(F_0(X), F_0(Y))$  along with a composition law $F_{\circ}(g,f)  : F_1(g) \circ F_1(f) = F_1(g \circ f) $ for all composable $g$ and $f$ and an identity law $F_{\idd}(X)   : \idd_{F_0(X)} = F_1(\idd_X)$ for all $X : \obb(\C)$.

We may refer to $F_0$ or $F_1$ by just $F$. We call $F_0$ and $F_1$ the \emph{$0$-} and \emph{$1$-dimensional} data of $F$, respectively. We call $F_{\circ}$ and $F_{\idd}$ its \emph{$2$-dimensional} data.

\item Let $F, G : \C \to \D$ be functors. A \emph{natural transformation} $\tau : F \Rightarrow G$ from $F$ to $G$ consists of functions 
$\tau_0 :\prod_{X : \obb(\C)}\homm_{\D}(F(X), G(X))$ and $\tau_1  : \prod_{X, Y : \obb(\C)}\prod_{f : \homm_{\C}(X , Y)}G(f) \circ \tau_0(X) = \tau_0(Y) \circ F(f)$. (If $\D \equiv \U$, we may tacitly use the equivalent type $G(f) \circ \tau_0(X) \sim \tau_0(Y) \circ F(f)$ instead.)
We say $\tau$ is a \emph{natural isomorphism} if each $\tau_0(X)$ is an equivalence.  
\end{enumerate}
\end{defi}

\begin{defi}[Adjunction]\label{adjdef}
Let $L : \C \to \D$ and $R: \D \to \C$ be functors of wild categories.  A \emph{(wild) adjunction $L \dashv R$} is a family of equivalences
$\psi  : \homm_{\D}(L(A), X) \simeq \homm_{\C}(A,R(X))$ equipped with functions witnessing that $\psi$ is natural in $X$ and $A$, respectively:
\begin{align*}
\nat_{\cod} & \ : \ \prod_{A :\obb(\C)}\prod_{ X, Y : \obb(\D)}\prod_{g : \homm_{\D}(X,Y)}\prod_{h : \homm_{\D}(L(A), X)}R(g) \circ \psi(h) = \psi(g \circ h)
\\ \nat_{\dom} & \ : \ \prod_{Y : \obb(\D)}\prod_{ A, B : \obb(\C)}\prod_{f : \homm_{\C}(A,B)}\prod_{h : \homm_{\D}(L(B), Y)}\psi(h) \circ f = \psi(h \circ L(f))
\end{align*} 
\end{defi}
For each adjunction $L \dashv R$, we also have naturality squares
\[\begin{tikzcd}[column sep = 22]
	{\homm_{\C}(A, R(X))} & {\homm_{\C}(A, R(Y))} & {\homm_{\C}(B, R(Y))} & {\homm_{\C}(A, R(Y))} \\
	{\homm_{\D}(L(A), X)} & {\homm_{\D}(L(A), Y)} & {\homm_{\D}(L(B), Y)} & {\homm_{\D}(L(A),Y)}
	\arrow[""{name=0, anchor=center, inner sep=0}, "{{R(g) \circ {-}}}", from=1-1, to=1-2]
	\arrow["{{\psi^{-1}}}"', from=1-1, to=2-1]
	\arrow["{{\psi^{-1}}}", from=1-2, to=2-2]
	\arrow[""{name=1, anchor=center, inner sep=0}, "{{{-} \circ f}}", from=1-3, to=1-4]
	\arrow["{{\psi^{-1}}}"', from=1-3, to=2-3]
	\arrow["{{\psi^{-1}}}", from=1-4, to=2-4]
	\arrow[""{name=2, anchor=center, inner sep=0}, "{{g \circ {-}}}"', from=2-1, to=2-2]
	\arrow[""{name=3, anchor=center, inner sep=0}, "{{{-} \circ L(f)}}"', from=2-3, to=2-4]
	\arrow["{{\widetilde{\nat}_{\cod}(g)}}"{description}, draw=none, from=0, to=2]
	\arrow["{{\widetilde{\nat}_{\dom}(f)}}"{description}, draw=none, from=1, to=3]
\end{tikzcd}\]
Here, the right-hand square is defined as the top path fitting into the commuting square
\[
\begin{tikzcd}[column sep = huge]
	{\psi^{-1}(h) \circ L(f)} && {\psi^{-1}(h \circ f)} \\
	{\psi^{-1}(\psi(\psi^{-1}(h) \circ L(f)))} && {\psi^{-1}(\psi (\psi^{-1}(h)) \circ f)}
	\arrow["{{\widetilde{\nat}_{\dom}(f,h)}}", Rightarrow, no head, dashed, from=1-1, to=1-3]
	\arrow["{{\eta_{\psi}(\psi^{-1}(h) \circ L(f))}}"', equals, from=1-1, to=2-1]
	\arrow["{{\ap_{\psi^{-1}}(\nat_{\dom}(f, \psi^{-1}(h)))}}"', equals, from=2-1, to=2-3]
	\arrow["{{\ap_{\psi^{-1}}(\ap_{{-} \circ f}(\epsilon_{\psi}(h)))}}"', equals, from=2-3, to=1-3]
\end{tikzcd}
\]
where $\eta_{\psi}$ and $\epsilon_{\psi}$ are from the equivalence data of $\psi$. The left-hand square is defined similarly.
\begin{rem}
Our results will make no use of $F_{\idd}$ or $\nat_{\cod}$, so we could have omitted them.
\end{rem}
\noindent In particular, $\mleft(\psi^{-1} , \widetilde{\nat}_{\dom} \mright)$ is a natural isomorphism $\homm_{\C}({-}, R(Y)) \Rightarrow \homm_{\D}(L({-}), Y)$ for each $Y: \obb(\D)$. The terms $\nat_{\dom}$ and $\widetilde{\nat}_{\dom}$ are related by the following \emph{exchange law}.

\begin{lem}\label{exchadj}
Let $\mleft(\psi, \nat_{\cod}, \nat_{\dom}\mright) : L \dashv R$. For all $f : \homm_{\C}(A,B)$ and $v : \homm_{\D}(L(B), Y)$, we have a commuting square
\[\begin{tikzcd}
	{\psi^{-1}(\psi(v))\circ L(f)} && {\psi^{-1}( \psi(v) \circ f)} \\
	{v \circ L(f)} && {\psi^{-1}(\psi(v \circ L(f)))}
	\arrow["{\widetilde{\nat}_{\dom}(f, \psi(v))}", equals, from=1-1, to=1-3]
	\arrow[""{name=0, anchor=center, inner sep=0}, "{\ap_{{-} \circ L(f)}(\eta_{\psi}(v))}"', equals, from=1-1, to=2-1]
	\arrow[""{name=1, anchor=center, inner sep=0}, "{\ap_{\psi^{-1}}(\nat_{\dom}(f,v))}", equals, from=1-3, to=2-3]
	\arrow["{\eta_{\psi}( v \circ L(f))}"', equals, from=2-1, to=2-3]
	\arrow["{\mathsf{exch}(f, v)}"{description}, draw=none, from=0, to=1]
\end{tikzcd}\]
\end{lem}
\begin{proof} 
Letting $\mathsf{mid} \coloneqq \ap_{\psi^{-1}}(\nat_{\dom}(f,v))^{-1} \cdot \ap_{\psi^{-1}}(\ap_{{-} \circ f}(\ap_{\psi}(\eta_{\psi}(v))))^{-1}$, define $\mathsf{exch}(f, v)$ as the chain of paths
\[ 
\begin{tikzcd}
	{\widetilde{\nat}_{\dom}(f, \psi(v))} \\
	{\eta_{\psi}(\psi^{-1}(\psi(v)) \circ L(f))^{-1} \cdot \ap_{\psi^{-1}}(\nat_{\dom}(f, \psi^{-1}(\psi(v))))^{-1} \cdot \ap_{\psi^{-1}}(\ap_{{-} \circ f}(\epsilon_{\psi}(\psi(v))))} \\
	{\eta_{\psi}( \psi^{-1}(\psi(v)) \circ L(f))^{-1} \cdot \ap_{\psi }(\ap_{{-} \circ L(f)}(\eta_{\psi}(v))) \cdot \mathsf{mid} \cdot \ap_{\psi^{-1}}(\ap_{{-} \circ f}(\epsilon_{\psi}(\psi(v))))} \\
	{\ap_{{-} \circ L(f)}(\eta_{\psi}(v)) \cdot \eta_{\psi}( v \circ L(f))^{-1} \cdot \mathsf{mid} \cdot \ap_{\psi^{-1}}(\ap_{{-} \circ f}(\epsilon_{\psi}(\psi(v))))} \\
	{\ap_{{-} \circ L(f)}(\eta_{\psi}(v)) \cdot \eta_{\psi}( v \circ L(f))^{-1} \cdot \ap_{\psi^{-1}}(\nat_{\dom}(f,v))^{-1}}
	\arrow["{{{{\textit{by definition}}}}}", equals, nfold, from=1-1, to=2-1]
	\arrow["{{\textit{via homotopy naturality of $\nat_{\dom}(f, {-})$}}}", equals, nfold, from=2-1, to=3-1]
	\arrow["{{\textit{via homotopy naturality of $\eta_{\psi}$}}}", equals, nfold, from=3-1, to=4-1]
	\arrow["{{\textit{via the triangle identity for $\psi$}}}", equals, nfold, from=4-1, to=5-1]
\end{tikzcd}
\]
\end{proof}

\subsection*{Limits}

Since we are studying colimits, we need to discuss cocones under diagrams. In HoTT, we have a concrete description of limits in the wild category of types that offers a useful way of representing cocones in general wild categories.
To avoid an infinite tower of coherence conditions (an unsolved problem in HoTT), we only consider diagrams over \emph{graphs}, which are type-theoretic versions of free categories~\cite[Section 3.2]{CSL25}. Let $\U$ be a universe. A \emph{graph} $\Gamma$ is a pair $\mleft(\Gamma_0, \Gamma_1\mright)$ consisting of a type $\Gamma_0 : \U$ of vertices and a family $\Gamma_1 : \Gamma_0 \to \Gamma_0 \to \U$ of edges.
 Given a wild category $\C$, a \emph{$\Gamma$-shaped diagram $F$ in $\C$} is a pair $\mleft(F_0, F_1\mright)$ consisting of a function $F_0 : \Gamma_0 \to \obb(\C)$ and a family of maps $F_1 : \prod_{i,j : \Gamma_0}\prod_{g : \Gamma_1(i,j)}  \homm_{\C}(F_0(i) , F_0(j))$. We may write $F$ for $F_0$ and $F_1$. A natural transformation between two diagrams is defined similarly to one between two functors.

Let $\Gamma$ be a graph. For every $\U$-valued diagram $F$ over $\Gamma$, the \emph{(standard) limit of $F$}~\cite[Definition 4.2.7]{Avi} is the type $\limm(F)  \coloneqq \sum_{\delta : \prod_{i : \Gamma_0}F_i}\prod_{i,j : \Gamma_0}\prod_{g : \Gamma_1(i,j)} F_{i,j,g}(\delta_i) = \delta_j$. The limit is functorial in $F$. The action on maps sends $\tau : F \Rightarrow G$ to the function $\limm(\tau) : \limm(F) \to \limm(G)$ defined by $\mleft(\delta, D\mright) \mapsto \mleft(\lambda{i}.\tau_0(i, \delta_i), \lambda{i}\lambda{j}\lambda{g}.\tau_1(i,j,g, \delta_i) \cdot \ap_{\tau_0(j)}(D_{i,j,g}) \mright)$.
Let $\C$ be a wild category. For every $\C$-valued diagram $F$ over $\Gamma$ and every $C : \obb(\C)$, we have the diagram $\homm_{\C}(F_j, C) \xrightarrow{{-} \circ F_{i,j,g}} \homm_{\C}(F_i, C)$ over $\Gamma^{\op}$, the opposite graph of $\Gamma$. We define the type of \emph{cocones under $F$ on $C$} as $\limm_{i : \Gamma^{\op}}(\homm_{\C}(F_i , C))$.
\begin{lem} \label{limpreseq}
Let $\Gamma$ be a graph and let $F$ and $G$ be $\Gamma$-shaped diagrams in $\U$. If $\tau : F \Rightarrow G$ is a natural isomorphism, then $\limm(\tau) : \limm(F) \to \limm(G)$ is an equivalence. 
\end{lem}

Next, we state the \emph{structure identity principle (SIP)} for $\limm$. The SIP is a general lemma characterizing identity types of $\Sigma$-types~\cite[Theorem 11.6.2]{FTID}. We will need the SIP for limits in order to port the proof of \emph{LAPC}.

\begin{lem}\label{SIPlim}
Let $F$ be a $\Gamma$-shaped diagram in $\U$. Let $e_1 \coloneqq \mleft(\delta_1, D_1\mright), e_2 \coloneqq \mleft(\delta_2, D_2\mright) : \limm(F)$. Then $e_1 = e_2$ is equivalent to the type of $Q : \delta_1 \sim \delta_2$ equipped with a commuting square 
\[\begin{tikzcd}[column sep = large]
	{F_{i,j,g}(\delta_1(i))} & {\delta_1(j)} \\
	{F_{i,j,g}(\delta_2(i))} & {\delta_2(j)}
	\arrow["{D_1(i,j,g)}", equals, from=1-1, to=1-2]
	\arrow["{\ap_{F_{i,j,g}}(Q_i)}"', equals, from=1-1, to=2-1]
	\arrow["{Q_j}", equals, from=1-2, to=2-2]
	\arrow["{D_2(i,j,g)}"', equals, from=2-1, to=2-2]
\end{tikzcd}\]
for each edge $g : \Gamma_1(i,j)$.
\end{lem}

\section{Porting the proof of \emph{LAPC}} \label{2Adj}

This section is the core of the paper. We find a sufficient, practically useful condition for the standard proof of \emph{LAPC} to work for wild categories. Informally, the condition states that $\nat_{\dom}$ interacts nicely with the composition law of the left adjoint.

Let $\C$ be a wild category. Let $\Gamma$ be a graph and $F : \Gamma \to \C$ be a $\C$-valued diagram over $\Gamma$. Consider a cocone $\K \coloneqq \mleft(C, r, K\mright)$ under $F$, where $r_i : F_i \to C$ for each $i : \Gamma_0$ and $K_{i,j,g} :  r_j \circ F_{i,j,g}  = r_i$ for all $i, j : \Gamma_0$ and $g : \Gamma_1(i,j)$.
\begin{defi}[{\cite[\myhref{https://github.com/PHart3/colimits-agda/blob/lapc-lmcs/HoTT-Agda/core/lib/wild-cats/Colim-wc.agda\#L22}{is-colim}]{2cohagda-pub}}] \label{colimwildc}
We say that $\K$ is \emph{colimiting} if for every $X : \obb(\C)$, the following post-composition map is an equivalence:
\begin{align*}
& \postcomp_{\K}(X) \ : \ \homm_{\C}(C, X) \to \limm_{i : \Gamma^{\op}}(\homm_{\C}(F_i, X))
\\ & \postcomp_{\K}(X,f) \ \coloneqq \ \mleft(\lambda{i}.f \circ r_i, \lambda{j}\lambda{i}\lambda{g}.\assoc(f, r_j, F_{i,j,g}) \cdot \ap_{f \circ {-}}(K_{i,j,g})  \mright)
\end{align*}
\end{defi}
\autoref{colimwildc} expresses that for every cocone $\K'$ under $F$, there is a unique cocone morphism $\K \to \K'$. Let $\D$ be a wild category and $L : \C \to \D$ be a functor. We have an induced diagram $L(F)$ and an induced cocone $L(\K)$ under $L(F)$:
\[\begin{tikzcd}[row sep = 25]
	{L(F_i)} && {L(F_j)} \\
	& {L(C)}
	\arrow[""{name=0, anchor=center, inner sep=0}, "{L(F_{i,j,g})}", from=1-1, to=1-3]
	\arrow["{L(r_i)}"', from=1-1, to=2-2]
	\arrow["{L(r_j)}", from=1-3, to=2-2]
	\arrow["{L(K_{i,j,g})}"{description}, draw=none, from=0, to=2-2]
\end{tikzcd}
\tag{$L(K_{i,j,g}) \coloneqq L_{\circ}(r_j, F_{i,j,g}) \cdot \ap_L(K_{i,j,g})$}\] 

Suppose that $\K$ is colimiting and that we have an adjunction $\mleft(\psi, \nat_{\cod}, \nat_{\dom}\mright) : L \dashv R$. We wish to replay the standard classical proof by showing the chain of isomorphisms \eqref{eq:iso} equals the post-composition map. Let us preview this argument.
Let $\zeta$ denote the composite of the isomorphisms. By function extensionality, it suffices to show $\zeta$ and post-composition are equal on $h$ for every $h : \homm_{\D}(L(C), Y)$. For each $i : \Gamma_0$, we can build an equality $Q_i : \psi^{-1}(\psi(h) \circ r_i) = h \circ L(r_i)$ from $\nat_{\dom}$ and the equivalence data for $\psi$. By the SIP for $\limm$ (\autoref{SIPlim}), it suffices to fill the following pentagon for all $i,j : \Gamma_0$ and $g : \Gamma_1(i,j)$:
\[\begin{tikzcd}
	& {\psi^{-1}(\psi(h) \circ r_i)} & \\
	{\psi^{-1}(\psi(h) \circ r_j) \circ L(F_{i,j,g})} && {h \circ L(r_i)} \\
	{\mleft(h \circ L(r_j)\mright) \circ L(F_{i,j,g})} && {h \circ \mleft(L(r_j) \circ L(F_{i,j,g})\mright)}
	\arrow["{Q_i}", equals, nfold, from=1-2, to=2-3]
	\arrow["{\pr_2(\zeta(h))(j ,i,g)}", equals, nfold, from=2-1, to=1-2]
	\arrow["{\ap_{{-} \circ L(F_{i,j,g})}(Q_j)}"', equals, nfold, from=2-1, to=3-1]
	\arrow["{\ap_{h \circ {-}}(L(K_{i,j,g}))}", equals, nfold, from=2-3, to=3-3]
	\arrow["{\assoc(h, L(r_j), L(F_{i,j,g}))}"', equals, nfold, from=3-1, to=3-3]
\end{tikzcd}\]
If $\D$'s hom types happened to be sets, like those of $1$-categories, then the equality filling it would hold trivially. The problem is that it need not hold for general $\D$, and we offer an example of a wild adjunction for which it is provably false inside HoTT (\autoref{eq-counter}). Still, by reverse engineering the equality, we arrive at the following general property of an adjunction as a sufficient condition for it to hold.

\begin{defi}[{\cite[\myhref{https://github.com/PHart3/colimits-agda/blob/lapc-lmcs/HoTT-Agda/core/lib/wild-cats/Ladj-2-coher.agda\#L16}{ladj-is-2coher}]{2cohagda-pub}}] \label{2coherdef}
The left adjoint $L$ is \emph{$2$-coherent} if for all $h_1 : \homm_{\D}(L(X), Y)$, $h_2 : \homm_{\C}(Z, X)$, and $h_3 : \homm_{\C}(W, Z)$, the following diagram commutes:
  \[
    \label{2coher}
    \begin{tikzcd}[column sep = huge, row sep = 25]
      {\mleft( \psi(h_1) \circ h_2\mright) \circ h_3} && {\psi(h_1) \circ \mleft(h_2 \circ h_3\mright)} \\
      {\psi(h_1 \circ L(h_2)) \circ h_3} && {\psi(h_1 \circ L(h_2 \circ h_3))} \\
      {\psi(\mleft(h_1 \circ L(h_2)\mright) \circ L(h_3))} && {\psi(h_1 \circ \mleft(L(h_2) \circ L(h_3) \mright))}
      \arrow["{{\assoc(\psi(h_1), h_2, h_3)}}", Rightarrow, no head, from=1-1, to=1-3]
      \arrow["{{\nat_{\dom}(h_2 \circ h_3, h_1)}}", Rightarrow, no head, from=1-3, to=2-3]
      \arrow["{{\ap_{{-} \circ h_3}(\nat_{\dom}(h_2,h_1))}}", Rightarrow, no head, from=2-1, to=1-1]
      \arrow["{{\ap_{\psi}(\ap_{h_1 \circ {-}}(L_{\circ}(h_2,h_3)))}}", Rightarrow, no head, from=2-3, to=3-3]
      \arrow["{{\nat_{\dom}(h_3,h_1 \circ L(h_2))}}", Rightarrow, no head, from=3-1, to=2-1]
      \arrow["{{\ap_{\psi}(\assoc(h_1, L(h_2), L(h_3)))}}"', Rightarrow, no head, from=3-1, to=3-3]
    \end{tikzcd} \tag{$\texttt{$2$-coh}$} \]
\end{defi}

\begin{note}
In terms of a classical biadjunction~\cite[Section 6]{kellybiadj}, the condition \eqref{2coher} is part of the pseudonaturality of the transformation $\mleft(\psi ,\nat_{\dom} \mright) : \homm_{\D}(L({-}), Y) \Rightarrow \homm_{\C}({-}, R(Y))$, expressing its compatibility with the relevant composition laws.
\end{note}

Intuitively, \autoref{2coherdef} accounts for the $2$-dimensional datum introduced by $L(K_{i,j,g})$. This is necessary as the definition of adjunction only accounts for $0$- and $1$-dimensional data. By accounting for the relevant $2$-dimensional data, we can finish porting the standard proof to wild category theory as follows.
\begin{thm}[{\cite[\myhref{https://github.com/PHart3/colimits-agda/blob/lapc-lmcs/HoTT-Agda/core/lib/wild-cats/Ladj-colim.agda}{Ladj-colim}]{2cohagda-pub}}] \label{LAPC}
If $L$ is $2$-coherent, the cocone $L(\K)$ under $L(F)$ is colimiting. 
\end{thm}
\begin{proof}
For all $i : \Gamma_0$ and $X : \obb(\D)$, we have
\begin{align*}
& \ \homm_{\D}(L(C), X)
\\ \simeq & \ \homm_{\C}(C, R(X)) \tag{$\textit{hom-isomorphism}$}
\\ \simeq & \ \limm_{i : \Gamma^{\op}}(\homm_{\C}(F_i, R(X))) \tag{$\textit{$\K$ is colimiting}$}
\\ \simeq & \ \limm_{i : \Gamma^{\op}}(\homm_{\D}(L(F_i), X)) \tag{$\textit{\autoref{limpreseq} applied to $\mleft(\psi^{-1} , \widetilde{\nat}_{\dom} \mright) \circ F$}$}
\end{align*}
Let $\zeta$ denote the composite of these three equivalences: $\zeta$ sends $h : \homm_{\D}(L(C), X)$ to the cocone $\mleft( \lambda{i}.\psi^{-1}(\psi(h) \circ r_i) , \textsf{comp-tri} \mright)$ where $\textsf{comp-tri}(j,i,g)$ denotes the chain of paths 
\[\begin{tikzcd}[column sep = 54]
	{\psi^{-1}(\psi(h) \circ r_j) \circ L(F_{i,j,g})} & {\psi^{-1}(\mleft(\psi(h) \circ r_j\mright) \circ F_{i,j,g})} & {\psi^{-1}(\psi(h) \circ r_i)}
	\arrow["{\widetilde{\nat}_{\dom}(F_{i,j,g}, \psi(h) \circ r_j)}", shift left, equals, nfold, from=1-1, to=1-2]
	\arrow["{\ap_{\psi^{-1}}(\assoc(\psi(h), r_j, F_{i,j,g}) \cdot \ap_{\psi(h) \circ {-}}(K_{i,j,g}) )}", shift left, equals, nfold, from=1-2, to=1-3]
\end{tikzcd}\]
We want to show $\zeta(h) = \postcomp_{L(\K)}(X, h)$. For each $i : \Gamma_0$, we have the chain $Q_i$ of paths
\[\begin{tikzcd}[column sep = 70]
	{\psi^{-1}(\psi(h) \circ r_i)} & {\psi^{-1}(\psi( h \circ  L(r_i)))} & {h \circ L(r_i)}
	\arrow["{\ap_{\psi^{-1}}(\nat_{\dom}(r_i, h))}", equals, from=1-1, to=1-2]
	\arrow["{\eta_{\psi}(h \circ L(r_i))}", equals, from=1-2, to=1-3]
\end{tikzcd}\]
Let $i, j : \Gamma_0$ and $g : \Gamma_1(i,j)$. By \autoref{SIPlim}, it suffices to show that
\[
\label{deseql}
\begin{tikzcd}
	{\widetilde{\nat}_{\dom}(F_{i,j,g}, \psi(h) \circ r_j) \cdot \ap_{\psi^{-1}}(\assoc(\psi(h), r_j, F_{i,j,g}) \cdot \ap_{\psi(h) \circ {-}}(K_{i,j,g}) ) \cdot Q_i} \\
	{\ap_{{-} \circ L(F_{i,j,g})}(Q_j) \cdot \assoc(h, L(r_j), L(F_{i,j,g})) \cdot \ap_{h \circ {-}}(L(K_{i,j,g}))}
	\arrow[Rightarrow, no head, from=1-1, to=2-1]
\end{tikzcd}
\tag{$\texttt{eq-edge}$} \]
We start with the top endpoint of \eqref{deseql}. By \autoref{exchadj}, we have the commuting diagram
\[ 
\begin{tikzcd}[column sep = 35, row sep = 30]
	{\psi^{-1}(\psi(h) \circ r_j) \circ L(F_{i,j,g})} && {\psi^{-1}(\mleft(\psi(h) \circ r_j\mright) \circ F_{i,j,g})} \\
	{\psi^{-1}(\psi(h \circ L(r_j))) \circ L(F_{i,j,g})} & \bullet & {\psi^{-1}(\psi(h \circ L(r_j)) \circ F_{i,j,g})} \\
	{\mleft(h \circ L(r_j)\mright) \circ L(F_{i,j,g})} && {\psi^{-1}(\psi(\mleft(h \circ L(r_j)\mright) \circ L(F_{i,j,g})))}
	\arrow[""{name=0, anchor=center, inner sep=0}, "{{\widetilde{\nat}_{\dom}(F_{i,j,g}, \psi(h) \circ r_j)}}", Rightarrow, no head, from=1-1, to=1-3]
	\arrow[Rightarrow, no head, from=1-1, to=2-1]
	\arrow[Rightarrow, no head, from=1-3, to=2-3]
	\arrow[""{name=1, anchor=center, inner sep=0}, "{{\widetilde{\nat}_{\dom}(F_{i,j,g}, \psi(h \circ L(r_j)))}}"', Rightarrow, no head, from=2-1, to=2-3]
	\arrow["{\textit{via $\eta_{\psi}(h \circ L(r_j))$}}"', Rightarrow, no head, from=2-1, to=3-1]
	\arrow["{\textit{via $\nat_{\dom}(F_{i,j,g},h \circ L(r_j))$}}", Rightarrow, no head, from=2-3, to=3-3]
	\arrow[""{name=2, anchor=center, inner sep=0}, "{{\eta_{\psi}( \mleft(h \circ L(r_j)\mright) \circ L(F_{i,j,g}))}}"', Rightarrow, no head, from=3-1, to=3-3]
	\arrow["{{\textit{homotopy naturality on $\nat_{\dom}(r_j,h)$}}}"{description}, draw=none, from=0, to=1]
	\arrow["{{\mathsf{exch}(F_{i,j,g}, h \circ L(r_j))}}"{description}, shift left=3, draw=none, from=2-2, to=2]
\end{tikzcd}\]
We also have the commuting diagram
\[\begin{tikzcd}[column sep = 52, row sep = 30]
	{\psi^{-1}(\psi(h) \circ \mleft(r_j \circ F_{i,j,g}\mright))} && {\psi^{-1}(\psi(h) \circ r_i)} \\
	{\psi^{-1}(\psi(h \circ L(r_j \circ F_{i,j,g})))} && {\psi^{-1}(\psi(h \circ L(r_i)))} \\
	{h \circ L(r_j \circ F_{i,j,g})} && {h \circ L(r_i)}
	\arrow[""{name=0, anchor=center, inner sep=0}, "{\ap_{\psi^{-1}}(\ap_{\psi(h) \circ {-}}(K_{i,j,g}))}", Rightarrow, no head, from=1-1, to=1-3]
	\arrow["{\ap_{\psi^{-1}}(\nat_{\dom}(r_j \circ F_{i,j,g}, h))}"', Rightarrow, no head, from=1-1, to=2-1]
	\arrow["{\ap_{\psi^{-1}}(\nat_{\dom}(r_i, h))}", Rightarrow, no head, from=1-3, to=2-3]
	\arrow[""{name=1, anchor=center, inner sep=0}, "{\ap_{\psi^{-1}}(\ap_{\psi}(\ap_{h \circ L({-})}(K_{i,j,g})))}"', Rightarrow, no head, from=2-1, to=2-3]
	\arrow["{\eta_{\psi}(h \circ L(r_j \circ F_{i,j,g}))}"', Rightarrow, no head, from=2-1, to=3-1]
	\arrow["{\eta_{\psi}(h \circ L(r_i))}", Rightarrow, no head, from=2-3, to=3-3]
	\arrow[""{name=2, anchor=center, inner sep=0}, "{\ap_{h \circ {-}}(\ap_L(K_{i,j,g}))}"', Rightarrow, no head, from=3-1, to=3-3]
	\arrow["{\textit{homotopy naturality}}"{description}, draw=none, from=0, to=1]
	\arrow["{\textit{homotopy naturality}}"{description}, draw=none, from=1, to=2]
\end{tikzcd}\]
By rewriting \eqref{deseql} with these two commuting diagrams, we turn it into the equality expressing that the following diagram commutes:
\[ \label{eq:untract}
\begin{tikzcd}[column sep = 34]
	{\psi^{-1}(\mleft( \psi(h) \circ r_j\mright) \circ F_{i,j,g})} && {\psi^{-1}(\psi(h) \circ \mleft(r_j \circ F_{i,j,g}\mright))} \\
	{\psi^{-1}(\psi(h \circ L(r_j)) \circ F_{i,j,g})} && {\psi^{-1}(\psi(h \circ L(r_j \circ F_{i,j,g})))} \\
	{\psi^{-1}(\psi(\mleft(h \circ L(r_j)\mright) \circ L(F_{i,j,g})))} && {h \circ L(r_j \circ F_{i,j,g})} \\
	{\mleft(h \circ L(r_j)\mright) \circ L(F_{i,j,g})} && {h \circ \mleft(L(r_j) \circ L(F_{i,j,g})\mright)}
	\arrow["{\ap_{\psi^{-1}}(\assoc(\psi(h), r_j, F_{i,j,g}))}", Rightarrow, shift left, no head, from=1-1, to=1-3]
	\arrow["{\ap_{\psi^{-1}}(\nat_{\dom}(r_j \circ F_{i,j,g}, h))}", Rightarrow, no head, from=1-3, to=2-3]
	\arrow["{\ap_{\psi^{-1}({-} \circ F_{i,j,g})}(\nat_{\dom}(r_j,h))}", Rightarrow, no head, from=2-1, to=1-1]
	\arrow["{\eta_{\psi}(h \circ L(r_j \circ F_{i,j,g}))}", Rightarrow, no head, from=2-3, to=3-3]
	\arrow["{\ap_{\psi^{-1}}(\nat_{\dom}(F_{i,j,g},h \circ L(r_j)))}", Rightarrow, no head, from=3-1, to=2-1]
	\arrow["{\eta_{\psi}( \mleft(h \circ L(r_j)\mright) \circ L(F_{i,j,g}))}"', Rightarrow, no head, from=3-1, to=4-1]
	\arrow["{\ap_{h \circ {-}}(L_{\circ}(r_j,F_{i,j,g}))}", Rightarrow, no head, from=3-3, to=4-3]
	\arrow["{\assoc(h, L(r_j), L(F_{i,j,g})) }"', Rightarrow, no head, from=4-1, to=4-3]
\end{tikzcd} \tag{$\dagger$} \]
At this point, we could have defined \eqref{2coher} so that it aligns with \eqref{eq:untract}, but this equality is difficult to check in practice. Hence we now transform \eqref{eq:untract} into something more tractable. Since $\psi$ is an equivalence (hence an embedding), the equality \eqref{eq:untract} is equivalent to its image under $\ap_{\psi}$. By homotopy naturality of $\eta_{\psi}$, this image is equivalent to an equality expressing that the following diagram commutes:
\[\begin{tikzcd}[column sep = 50]
	{\mleft( \psi(h) \circ r_j\mright) \circ F_{i,j,g}} && {\psi(h) \circ \mleft(r_j \circ F_{i,j,g}\mright)} \\
	{\psi(h \circ L(r_j)) \circ F_{i,j,g}} && {\psi(h \circ L(r_j \circ F_{i,j,g}))} \\
	{\psi(\mleft(h \circ L(r_j)\mright) \circ L(F_{i,j,g}))} && {\psi(h \circ \mleft(L(r_j) \circ L(F_{i,j,g}) \mright))}
	\arrow["{\assoc(\psi(h), r_j, F_{i,j,g})}", Rightarrow, no head, from=1-1, to=1-3]
	\arrow["{\nat_{\dom}(r_j \circ F_{i,j,g}, h)}", Rightarrow, no head, from=1-3, to=2-3]
	\arrow["{\ap_{{-} \circ F_{i,j,g}}(\nat_{\dom}(r_j,h))}", Rightarrow, no head, from=2-1, to=1-1]
	\arrow["{\ap_{\psi}(\ap_{h \circ {-}}(L_{\circ}(r_j,F_{i,j,g})))}", Rightarrow, no head, from=2-3, to=3-3]
	\arrow["{\nat_{\dom}(F_{i,j,g},h \circ L(r_j))}", Rightarrow, no head, from=3-1, to=2-1]
	\arrow["{\ap_{\psi}(\assoc(h, L(r_j), L(F_{i,j,g})))}"', Rightarrow, no head, from=3-1, to=3-3]
\end{tikzcd}\] Finally, this diagram commutes because $L$ is $2$-coherent.
\end{proof}

\begin{exa} \label{eq-counter}
As claimed, without the $2$-coherence assumption, \eqref{deseql} may not hold. 

Define the wild category $\E$ by $\obb(\E) \coloneqq \1$ and $\hom_{\E}(\ast, \ast) \coloneqq S^1$. Here, $S^1$ denotes the circle, defined as the HIT generated by 
a point $\base : S^1$ and a path $\loopp : \base = \base$. The remaining structure on $\E$, including its composition operation $\bullet$, comes from path concatenation on a loop space (where the loop space $\Omega(X,x_0)$ of a pointed type is the pointed type $\mleft(x_0=x_0, \refl_{x_0}\mright)$). Indeed, $S^1$ is the loop space of the Eilenberg-MacLane space $K(\Z, 2)$~\cite{EMS}. For all $\ell : \hom_{\E}(\ast, \ast)$, we have a nontrivial loop $\L_{\ell}$ at $\ell$: It is known that $\loopp$ is nontrivial~\cite[Lemma 6.4.1]{Uni13}, and we have an equivalence $f : \mleft(S^1, \base \mright) \to \mleft(S^1 , \ell\mright)$ of pointed types defined by $f(x) \coloneqq x \bullet \ell$. (As we'll see in \cref{SuspCol}, this equivalence illustrates how loop spaces are \emph{homogeneous}, i.e., independent of basepoint in this sense.) The induced map $\Omega(f) : \Omega(S^1, \base) \to_{\ast} \Omega(S^1,\ell) $ on loop spaces is also an equivalence because $\Omega$ is functorial, and we define $\L_{\ell}$ as the image of $\loopp$ under $\Omega(f)$.
Let the functor $\Lambda : \E \to \E$ be the identity on objects and morphisms, but let $\Lambda_{\circ}(\ell_1, \ell_2) \coloneqq \L_{\ell_1 \bullet \ell_2}$.
Consider the evident adjunction $\Lambda \dashv \Lambda$. If $h \equiv \idd_{\ast}$, then \eqref{deseql}, with respect to this adjunction, reduces to showing that $\Lambda_{\circ} (h_2, h_3)$ is trivial. But it's nontrivial by construction.
\end{exa}

Although \autoref{eq-counter} shows that the standard proof fails, it does not show that wild left adjoints fail to preserve colimits. Indeed, we now prove that $\Lambda$ preserves colimits. For every $\E$-valued diagram $F$ over a graph $\Gamma$, the type of cocones under $F$ is equivalent to  $\mathcal{T}_F \coloneqq \sum_{h : \Gamma_0 \to S^1}\prod_{i,j,g}F_{i,j,g} \bullet h_j = h_i$. Let $\K \coloneqq \mleft(h,K\mright)$ be a cocone in this form and consider its image under $\Lambda$: $\mleft(h, \lambda{j}\lambda{i}\lambda{g}.\L_{F_{i,j,g} \bullet h_j} \cdot K_{i,j,g} \mright)$. It suffices to show that the following triangle commutes where $e$ is defined by $e(f, M) \coloneqq \mleft(f,\lambda{j}\lambda{i}\lambda{g}.\L_{F_{i,j,g} \bullet f_j} \cdot M_{i,j,g} \mright)$ with inverse sending $\mleft(f,M\mright)$ to $\mleft(f,\lambda{j}\lambda{i}\lambda{g}.\L_{F_{i,j,g} \bullet f_j}^{-1} \cdot M_{i,j,g}  \mright)$:
\[
\begin{tikzcd}
	& {S^1} & \\
	{\mathcal{T}_F} && {\mathcal{T}_F}
	\arrow["{\postcomp_{\K}}"', from=1-2, to=2-1]
	\arrow["{\postcomp_{\Lambda(\K)}}", from=1-2, to=2-3]
	\arrow["e"', from=2-1, to=2-3]
	\arrow["\simeq", from=2-1, to=2-3]
\end{tikzcd}
\] This triangle commutes as follows: For each $x : S^1$,
\begin{align*}
& \ e(\postcomp_{\K}(x)) 
\\  \equiv & \   \mleft(h \bullet x, \lambda{j}\lambda{i}\lambda{g}.\L_{F_{i,j,g}\bullet \mleft(h_j\bullet x\mright)} \cdot \assoc(x,h_j, F_{i,j,g}) \cdot \ap_{{-} \bullet x}(K_{i,j,g})\mright)
\\ = & \ \mleft(h \bullet x, \lambda{j}\lambda{i}\lambda{g}. \assoc(x,h_j, F_{i,j,g}) \cdot \L_{\mleft(F_{i,j,g}\bullet h_j\mright)\bullet x} \cdot \ap_{{-} \bullet x}(K_{i,j,g})\mright) \tag{$\textit{homotopy naturality}$}
\\ = & \ \mleft(h \bullet x, \lambda{j}\lambda{i}\lambda{g}.\assoc(x,h_j, F_{i,j,g}) \cdot \ap_{{-} \bullet x}(\L_{F_{i,j,g} \bullet h_j} \cdot K_{i,j,g})\mright)
\\ \equiv & \ \postcomp_{\Lambda(\K)}(x) 
\end{align*}
The final identity follows from the fact that $\L_{y \bullet x} = \ap_{{-}\bullet x}(\L_y)$ for all $y : S^1$, which holds because $\L_{y} \equiv \ap_{{-} \bullet y}(\loopp)$.

Thus, \autoref{eq-counter} tells us that $2$-coherence is sufficient but not necessary for preservation of colimits. Importantly, we still lack an example of a wild left adjoint that fails to preserve colimits.

\begin{exa}[{\cite[\myhref{https://github.com/PHart3/colimits-agda/blob/lapc-lmcs/HoTT-Agda/core/lib/wild-cats/LAPC-counterexmp.agda}{LAPC-counterexmp}]{2cohagda-pub}}] \label{truecexmp}
We fill this gap with a slightly more complex example. Define $\E_{\textit{new}}$ by $\obb(\E_{\textit{new}}) \coloneqq \2$ and $\homm_{\E_{\textit{new}}}(0,0) \coloneqq \1$, $\homm_{\E_{\textit{new}}}(0,1) \coloneqq S^1$, $\homm_{\E_{\textit{new}}}(1,0) \coloneqq \1$, and $\homm_{\E_{\textit{new}}}(1,1) \coloneqq S^1 \to S^1$. Define composition of morphisms as follows.
\begin{itemize} 
\item For all $g, f : \homm_{\E_{\textit{new}}}(1,1)$, define $g \circ f$ as the usual composite of functions.
\item For all $h : \homm_{\E_{\textit{new}}}(1,1)$ and $x : \homm_{\E_{\textit{new}}}(0,1)$, let  $h \circ x \coloneqq h(x)$.
\item For all $y : \homm_{\E_{\textit{new}}}(0,1)$ and $x : \homm_{\E_{\textit{new}}}(1,0)$, define $y \circ x$ as the constant map at $y$. 
\item Take the obvious identity morphisms.
\end{itemize}
The unit and associativity laws hold definitionally after case splitting on $\2$. Next, recall from \cite[Lemma 6.4.2]{Uni13} the function $\tau : \prod_{x : S^1}x = x$ defined by $S^1$-induction with $\tau(\base) \equiv \loopp$. Define the graph $\Gamma$ with a single vertex $v$ and a single edge $\ell$ from $v$ to itself. Define the diagram $F : \Gamma \to \E_{\textit{new}}$ by $F(v) \coloneqq 0$ and $F(\ell) \coloneqq \idd_0$. Consider the cocone $\K \coloneqq \mleft(1, \base, \loopp\mright)$ under $F$:
\[\begin{tikzcd}
	0 && 0 \\
	& 1
	\arrow[""{name=0, anchor=center, inner sep=0}, "\ast", from=1-1, to=1-3]
	\arrow["{\base}"', from=1-1, to=2-2]
	\arrow["{\base}", from=1-3, to=2-2]
	\arrow["{\loopp}"{description}, draw=none, from=0, to=2-2]
\end{tikzcd}\] The function $\postcomp_{\K}(0) : \1 \to \1$ is clearly an equivalence. Further, by the universal property of $S^1$, the function $\postcomp_{\K}(1) : \mleft(S^1 \to S^1\mright) \to \sum_{x : S^1}x = x$, defined by $f \mapsto \mleft( f(\base), \ap_f(\loopp)\mright)$, is an equivalence. Hence $\K$ is colimiting.

We want to build a functor $\Lambda$ on $\E_{\textit{new}}$ that fails to take $\K$ to a colimiting cocone. As before, define $\Lambda : \E_{\textit{new}} \to \E_{\textit{new}}$ as the identity on objects and morphisms. But now define \[
\Lambda_{\circ}(g,f) \coloneqq 
\begin{cases}
\tau(g)^{-1}  & \text{$g : \homm_{\E_{\textit{new}}}(0, 1)$ and $f : \homm_{\E_{\textit{new}}}(0, 0)$} \\  
\refl_{g \circ f} & \text{otherwise}
\end{cases}
\] We still have an evident wild adjunction $\Lambda \dashv \Lambda$. Note that $\Lambda$ takes $\K$ to the following cocone under $F$: 
\[\begin{tikzcd}
	0 && 0 \\
	& 1
	\arrow[""{name=0, anchor=center, inner sep=0}, "\ast", from=1-1, to=1-3]
	\arrow["{\base}"', from=1-1, to=2-2]
	\arrow["{\base}", from=1-3, to=2-2]
	\arrow["{\kappa}"{description}, draw=none, from=0, to=2-2]
\end{tikzcd}\]
where $\kappa \coloneqq \Lambda_{\circ}(\base, \idd_0) \cdot \loopp : \base = \base$.
But $\Lambda_{\circ}(\base, \idd_0) \equiv \loopp^{-1}$, so that $\Lambda(\K) = \mleft(1, \base, \refl_{\base}\mright)$. We claim that this cocone is not colimiting. To see this, notice that $\postcomp_{\Lambda(\K)}(1) : \mleft(S^1 \to S^1\mright) \to \sum_{x : S^1}\mleft(x = x\mright)$ is defined by $f \mapsto \mleft(f(\base), \refl_{f(\base)}\mright)$. This is not an equivalence, because its fiber over $\mleft(\base, \loopp\mright)$ is empty. Indeed, an identity $\mleft(\base, \loopp\mright) = \mleft(f(\base), \refl\mright)$ for some $f$ would imply that $\loopp = \refl_{\base}$, which is false. 
\end{exa}

\begin{rem}
Kraus has used the family $\L$ of nontrivial loops, defined in \autoref{eq-counter}, to form a wild category that is provably not a bicategory~\cite[Lemma 8]{eating}, i.e., that fails to satisfy the triangle and pentagon identities. That said, the issue addressed by our two counterexamples---the cocontinuity of wild left adjoints---is orthogonal to Kraus's concern. Indeed, both $\E$ and $\E_{\textit{new}}$ are themselves bicategories, so imposing this extra structure does not fix our problem.
\end{rem}

\section{Suspension is $2$-coherent} \label{SuspCol}

This section (along with \autoref{ModCol}) offers evidence that \autoref{2coherdef} is useful in practice. We show that the suspension endofunctor $\Sigma : \U_{\ast} \to \U_{\ast}$ on the wild category of pointed types is a $2$-coherent left adjoint to the loop space endofunctor $\Omega$. By \autoref{LAPC}, we deduce that $\Sigma$ preserves (graph-indexed) colimits. Although the diagram \eqref{2coher} is a valuable approach to proving this preservation property, it does rely on a new trick based on \emph{homogeneous types} to handle the path algebra generated by $\Sigma$. The adjunction $\Sigma \dashv \Omega$ is already known~\cite[Lemma 2.16]{Hur}, so our contribution is verifying that $\Sigma$ is $2$-coherent. Also, to our knowledge, ours is the first proof in HoTT (in particular, Book HoTT) that $\Sigma$ preserves colimits.

If $X$ is a pointed type, then $\ty(X)$ and $\pt(X)$ denote the underlying type and basepoint of $X$, respectively. If $f : X \to_{\ast} Y$ is a pointed map, then $\fun(f)$ and $\bp(f)$ denote $f$'s underlying function and proof of basepoint preservation, respectively.
\begin{defi}
Let $f_1, f_2 : X_1 \to_{\ast} X_2$ be pointed maps. A \emph{pointed homotopy $f_1 \sim_{\ast} f_2$} is a homotopy $H : \fun(f_1) \sim \fun(f_2)$ together with a path $ H(\pt(X_1)) \cdot \bp(f_2) = \bp(f_1)$.
\end{defi}
The SIP for pointed maps says that the canonical function $f_1 = f_2 \to f_1 \sim_{\ast} f_2$ is an equivalence, whose  inverse we denote by $\mleft\langle{{-},{-}}\mright\rangle$. We use this equivalence to define the terms $\Sigma_{\circ}$ and $\nat_{\dom}$. Doing so will help us manipulate them in the proof of $2$-coherence. 

The composition law's first component is defined by induction (see \autoref{suspdef}), with $t_{\mathsf{N}}$ and $t_{\mathsf{S}}$ trivial and $T$ defined, for suitable pointed maps $r$ and $s$, via the commuting square
\[\begin{tikzcd}[column sep = 70]
	{\ap_{\Sigma(s \circ r)}(\glue(x))} & {\ap_{\Sigma(s) \circ  \Sigma(r)}(\glue(x))} \\
	{\glue(\fun(s)(\fun(r)(x)))} & {\ap_{\Sigma(s)}(\glue(\fun(r)(x)))}
	\arrow[equals, from=1-1, to=1-2]
	\arrow["{\beta_{\Sigma(s \circ r)}(x)}"', equals, from=1-1, to=2-1]
	\arrow["{\textit{via $\beta_{\Sigma(r)}(x)$}}", equals, from=1-2, to=2-2]
	\arrow["{\textit{via $\beta_{\Sigma(s)}(\fun(r)(x))$}}"', equals, from=2-1, to=2-2]
\end{tikzcd}\]
Its second component is trivial.
The adjunction $\Sigma \dashv \Omega$ is defined by
\begin{align*}
& \Phi \ : \  \homm_{\U_{\ast}}(\Sigma(X, x_0), \mleft(Y, y_0\mright)) \xrightarrow{\simeq} \homm_{\U_{\ast}}(\mleft(X, x_0\mright), \Omega(Y, y_0))
\\  & \Phi(h, h_0) \ \coloneqq \ (\lambda{x}.\underbrace{h_0^{-1} \cdot \ap_h(\glue(x) \cdot \glue(x_0)^{-1}) \cdot h_0}_{\xi(x,h, h_0)} , \zeta(\glue(x_0), h_0) ) 
\end{align*}
where the underbrace denotes term abbreviation and $\zeta(\glue(x_0), h_0)$ denotes the evident path of type $\xi(x_0,h,h_0)  = \refl_{y_0}$. For all $f^{\ast} \coloneqq \mleft(f, f_0\mright) : \mleft(Z, z_0\mright) \to_{\ast} \mleft(X, x_0\mright)$ and $h^{\ast} \coloneqq \mleft(h, h_0\mright) : \Sigma(X, x_0) \to_{\ast} \mleft(Y, y_0\mright)$, $\nat_{\dom}(f^{\ast}, h^{\ast})$ is defined as the path
\begin{align*}
& \ \Phi(h^{\ast}) \circ  f^{\ast}
\\ \equiv & \ \mleft(\lambda{x}.\xi(f(x), h^{\ast}), \ap_{\xi({-}, h^{\ast})}(f_0) \cdot \zeta(\glue(x_0), h_0)  \mright)
\\ = & \ \mleft(\lambda{x}.\xi(x,h \circ \fun(\Sigma(f^{\ast})), h_0), \zeta(\glue(z_0), h_0)  \mright) \tag{$\mleft\langle{\Theta , \Theta_0}\mright\rangle$}
\\  \equiv & \ \Phi(h \circ \fun(\Sigma(f^{\ast})), h_0) 
\\ \equiv & \ \Phi(h^{\ast} \circ \Sigma(f^{\ast}))
\end{align*}
Here, for each $z : Z$, $\Theta(z)$ is defined as the path
\[\begin{tikzcd}
	{h_0^{-1} \cdot \ap_h(\glue(f(z)) \cdot \glue(f(z_0))^{-1}) \cdot h_0} \\
	{h_0^{-1} \cdot \ap_{h \circ \fun(\Sigma(f^{\ast}))}(\glue(z) \cdot \glue(z_0)^{-1}) \cdot h_0}
	\arrow["{\textit{via $\beta_{\Sigma(f^{\ast})}(z)$ and $\beta_{\Sigma(f^{\ast})}(z_0)$}}", Rightarrow, no head, from=1-1, to=2-1]
\end{tikzcd}\] and $\Theta_0$ is defined by path induction on $\beta_{\Sigma(f^{\ast})}(z_0)$.

Now that we've defined $\Sigma_{\circ}$ and $\nat_{\dom}$, we claim that the diagram \eqref{2coher} commutes. The SIP for pointed homotopies turns this goal into a \emph{double pointed homotopy}:
\begin{defi}
  Let $f_1$ and $f_2$ be pointed maps and let $\mleft(H_1, \kappa_1\mright), \mleft(H_2, \kappa_2\mright) : f_1 \sim_{\ast} f_2$. A \emph{double pointed homotopy} $\mleft(H_1, \kappa_1\mright) \sim_{\ast}^2 \mleft(H_2, \kappa_2\mright)$ consists of a homotopy $\mu : H_1 \sim H_2$ and a commuting triangle 
\[\begin{tikzcd}
	{H_2(\pt(X_1)) \cdot \bp(f_2)} \\
	{H_1(\pt(X_1)) \cdot \bp(f_2)} & {\bp(f_1)}
	\arrow["{\ap_{{-} \cdot \bp(f_2)}(\mu(\pt(X_1)))}"', equals, from=1-1, to=2-1]
	\arrow["{\kappa_2}", equals, from=1-1, to=2-2]
	\arrow["{\kappa_1}"', equals, from=2-1, to=2-2]
\end{tikzcd}\]
\end{defi}

To construct the first component of the desired double pointed homotopy, we have to reduce a large expression involving various typal $\beta$-rules (coming from the $\nat_{\dom}$ and $\Sigma_{\circ}$ edges of \eqref{2coher}). We do so via a mechanical process of iteratively eliminating matching $\beta$-rules. The commuting triangle, however, is infeasible to construct directly. The problem is that it contains the entire first component, which involves complex path algebra and does not reduce at $\pt(W)$. In addition, it contains a handful of nontrivial path inductions from $\Theta$'s second component. The result is an expression that is simply too big.

Luckily, we can get the commuting triangle for free by noticing the special nature of loop spaces. Every loop space is a \emph{homogeneous type}, i.e., a pointed type $\mleft(X, x_0\mright)$ equipped with a pointed equivalence $\auto_x : \mleft(X, x_0\mright) \xrightarrow{\simeq}_{\ast} \mleft(X, x\mright)$ for every $x : X$. In this case, we also say $X$ is \emph{homogeneous at $x_0$}. We note a few things about such pointed types. First, if $M$ is homogeneous, then it is so at all its elements. Second, by applying ${-} \circ \auto_{\pt(M)}^{-1}$ to $\auto$, we can make every homogeneous type $M$ \emph{strongly homogeneous}, i.e., make $\auto$ satisfy $\auto_{\pt(M)} = \idd_M$. Finally, as $\Omega$ preserves pointed equivalences, if $\mleft(X, x_0\mright)$ is homogeneous, we have $\homogpth(x_0, x) : \mleft(x_0 = x_0\mright) \xrightarrow{\simeq} \mleft(x = x\mright)$ for every $x : X$. 
Now, a key insight for our goal is Cavallo's trick: two pointed maps into homogeneous types are pointed-homotopic when their underlying functions are homotopic~\cite[\myhref{https://github.com/agda/cubical/blob/master/Cubical/Foundations/Pointed/Homogeneous.agda\#L41}{$\mathtt{\rightarrow \! \! \cdot Homogeneous \! \! \equiv}$}]{cubical}.
As our goal is a higher pointed homotopy, we want the following higher version of the trick, which will finish the proof that $\Sigma$ is $2$-coherent~\cite[\myhref{https://github.com/PHart3/colimits-agda/blob/lapc-lmcs/HoTT-Agda/theorems/homotopy/Suspension/Susp-2coher.agda}{Susp-2coher}]{2cohagda-pub}:
\begin{lem}[{\cite[\myhref{https://github.com/PHart3/colimits-agda/blob/lapc-lmcs/HoTT-Agda/core/lib/types/Homogeneous.agda\#L151}{\text{$\sim \!\! \odot$homog$\sim$}}]{2cohagda-pub}}] \label{trick}
  Let $f_1, f_2 : X_1 \to_{\ast} X_2$ with $X_2$ homogeneous. Let $\mleft(H_1, \kappa_1\mright), \mleft(H_2, \kappa_2\mright) : f_1 \sim_{\ast} f_2$. If $H_1 \sim H_2$, then $\mleft(H_1, \kappa_1\mright) \sim_{\ast}^2 \mleft(H_2, \kappa_2\mright)$.
\end{lem}
\begin{proof}
We begin with a general observation. Let $k : X_1 \to_{\ast} X_2$ and $k_0 \coloneqq \fun(k)(\pt(X_1))$. Consider the evaluation map $\ev_{\pt(X_1),\fun(k)} : \mleft(\fun(k) \sim \fun(k), \refl \mright) \to_{\ast} \Omega(X_2, k_0)$.  As $X_2$ is homogeneous at $k_0$, this map has a pointed section $\sigma^{\ast}$ whose underlying function sends a loop $p$ at $k_0$ to the homotopy $\sigma(p,x) \coloneqq \homogpth(k_0, \fun(k)(x), p)$. It's easy to check that $\sigma$ is pointed. It remains to construct a pointed homotopy $\gamma : \ev_{\pt(X_1),\fun(k)} \circ \sigma^{\ast} \sim_{\ast} \idd$.  
The first component of $\gamma$ is a homotopy $ \homogpth(k_0, k_0) \sim \idd_{k_0 = k_0}$, which we get by promoting $X_2$ to a strongly homogeneous type.  The second component, which also uses the fact $X_2$ is strongly homogeneous, follows routinely.
 
Let $Q : H_1 \sim H_2$. 
By strong function extensionality, we ``path induct'' on $H_1$ and $Q$ so that they are both identity homotopies. Further, by generalizing $\pt(X_2)$, we induct on $\bp(f_1)$ to make it $\refl_{w_0}$ with $w_0 \coloneqq \fun(f_1)(\pt(X_1))$. By our general observation, $\ev_{\pt(X_1), \fun(f_1)}$ has a pointed section $\sigma^{\ast}$, so that $\Omega(\sigma^{\ast})$ is a pointed section of $\Omega(\ev_{\pt(X_1), \fun(f_1)})$. Now, we want a pair $\mleft(\mu, \mu_0\mright) : \mleft(\refl, \kappa_1\mright) \sim_{\ast}^2 \mleft(\refl, \kappa_2\mright)$. We define $\mu : \refl \sim \refl$ as the image $\fun(\Omega(\sigma^{\ast}))(\kappa)$ of a certain loop $\kappa$ at $\refl_{w_0}$ under $\Omega(\sigma^{\ast})$. To make the right choice for $\kappa$, we look ahead to the commuting triangle $\mu_0$:
\[\begin{tikzcd}
	{\bp(f_1)} \\
	{\bp(f_1)} & {\refl_{w_0}}
	\arrow["{\ap_{{-} \cdot \bp(f_1)}(\mu (\pt(X_1)))}"', equals, from=1-1, to=2-1]
	\arrow["{\kappa_2}", equals, from=1-1, to=2-2]
	\arrow["{\kappa_1}"', equals, from=2-1, to=2-2]
\end{tikzcd}\] Since $\Omega(\sigma^{\ast})$ is a pointed section of $\Omega(\ev_{\pt(X_1), \fun(f_1)})$, $\mu (\pt(X_1))$ will equal $\kappa$. Finally, ${-} \cdot \bp(f_1)$ is an equivalence, so we simply solve for $\kappa$.
\end{proof}

\begin{thm}[{\cite[\myhref{https://github.com/PHart3/colimits-agda/blob/lapc-lmcs/HoTT-Agda/theorems/homotopy/Suspension/Susp-colim.agda}{Susp-colim}]{2cohagda-pub}}]\label{suspcol}
The suspension $\Sigma : \U_{\ast} \to \U_{\ast}$ preserves colimits.
\end{thm}

\begin{rem}
One may try to derive \autoref{suspcol} from the property that pushouts commute with colimits as follows. For types, the $3\times 3$ lemma~\cite[Section VII]{3x3} tells us pushouts commute with pushouts, and preservation of coproducts by wild left adjoints is easy to prove since the underlying graph is discrete. If every colimit is a pushout of coproducts, we're done. But there are problems: This characterization of colimits \emph{within $\U_{\ast}$} is quite nontrivial on its own~\cite[Section 8.1]{CSL25}. Also, the $3\times 3$ lemma on its face says nothing about pushouts in $\U_{\ast}$.
\end{rem}

Recall from \cite{acyclic} that a type is \emph{acyclic} if its suspension is contractible. We contribute a new closure property of acyclic types with the next corollary to \autoref{suspcol}. It uses Hart and Favonia's construction $\colimm^{\ast}$ of colimits in $\U_{\ast}$ as the cofiber of a map between colimits in $\U$~\cite[Theorem 15]{CSL25}. (Colimits in $\U$ are postulated as HITs definable from pushouts.)
\begin{cor}[{\cite[\myhref{https://github.com/PHart3/colimits-agda/blob/lapc-lmcs/HoTT-Agda/theorems/homotopy/Acyc-colim.agda}{Acyc-colim}]{2cohagda-pub}}] 
The pointed acyclic types are closed under colimits in $\U_{\ast}$.
\end{cor}
\begin{proof}
By \autoref{suspcol} and uniqueness of colimits in $\U_{\ast}$ (as in any wild bicategory), we have an equivalence of pointed types $\Sigma(\colimm^{\ast}(F)) \simeq_{\ast} \colimm^{\ast}(\Sigma(F))$. If $\ty(F_i)$ is acyclic for each $i : \Gamma_0$, then $\ty(\colimm^{\ast}(\Sigma(F)))$ is contractible as the cofiber of an equivalence, namely the function $\colimm(\1) \to \colimm(\ty \circ \Sigma(F))$ (between colimits in $\U$) induced by the unique natural transformation into $\ty \circ \Sigma(F)$. This completes the proof since equivalences preserve contractibility.
\end{proof}

\begin{note} \label{joinpres}
We can improve \autoref{suspcol} as follows. For types $T$ and $U$, recall that the join of $T$ and $U$~\cite[Section 6.8]{Uni13} is the pushout
\[\begin{tikzcd}
	{T \times U} & U \\
	T & {T \ast U}
	\arrow["{\pr_2}", from=1-1, to=1-2]
	\arrow["{\pr_1}"', from=1-1, to=2-1]
	\arrow[from=1-2, to=2-2]
	\arrow[from=2-1, to=2-2]
	\arrow["\lrcorner"{anchor=center, pos=0.125, rotate=180}, draw=none, from=2-2, to=1-1]
\end{tikzcd}\]
Let $X$ be a pointed type. We have a wild functor $X \ast {-} : \U_{\ast} \to \U_{\ast}$ and an isomorphism between $\2 \ast {-}$ and $\Sigma$, so $X \ast {-}$ generalizes the suspension. We also have a wild adjunction $X \ast {-} \dashv X \to_{\ast} \Omega({-})$~\cite[\myhref{https://github.com/PHart3/colimits-agda/blob/lapc-lmcs/HoTT-Agda/theorems/homotopy/Join/JoinAdjointLoopCod.agda\#L286}{JoinLoopCodAdj}]{2cohagda-pub}, which has previously been observed and mechanized by Cagne et al.~\cite[Proposition B.10]{symsph}. If $X \equiv \2$, this adjunction reduces to the suspension-loop adjunction described above (in light of the isomorphism between $\2 \to _{\ast} {-}$ and the identity wild functor). As $X \to_{\ast} Y$, with basepoint the constant map, is homogeneous when $Y$ is, we can use a parameterized variant of the argument for \autoref{suspcol} to show that $X \ast {-}$ is $2$-coherent and thus preserves colimits~\cite[\myhref{https://github.com/PHart3/colimits-agda/blob/lapc-lmcs/HoTT-Agda/theorems/homotopy/Join/Join-colim.agda}{Join-colim}]{2cohagda-pub}.
\end{note}

\section{Colimits of modal types} \label{ModCol}

We end with another application of \autoref{LAPC} in synthetic homotopy theory. We prove that all \emph{modalities} on coslices of a universe $\U$ are $2$-coherent and thereby construct (graph-indexed) colimits of modal types. Consider functions $\modal : \U \to \U$ and $\eta : \prod_{X : \U}X \to \modal{X}$. A type $X : \U$ is \emph{modal} if $\eta_X$ is an equivalence. Let $\U_{\modal}$ denote the subuniverse of modal types.
\begin{defi}[{\cite[\myhref{https://github.com/PHart3/colimits-agda/blob/lapc-lmcs/HoTT-Agda/core/lib/types/Modality.agda\#L10}{Modality}]{2cohagda-pub}}]
We say that $\modal$ is a \emph{modality} if 
\begin{itemize}
\item for all $X : \U$, $\modal{X}$ is modal;
\item for all $X : \U$ and $x, y : \modal{X}$, the identity type $x = y$ is modal;
\item for all $X : \U$ and $P : \modal{X} \to \U_{\modal}$, the function ${-} \circ \eta_X : \mleft(\prod_{x : \modal{X}}P(x)\mright) \to \mleft(\prod_{x : X}P(\eta_X(x))\mright)$ has a section. (This condition is called \emph{$\modal$-induction}.)
\end{itemize}
\end{defi} 
 
Let $\modal$ be a modality and $A$ be a type. By $\modal$-induction (including the associated nondependent recursion principle), we have a functor $\modal^A : A/\U \to \mleft(A/\U\mright)_{\modal}$ into the full wild subcategory of $A/\U$ on those $\mleft(X,s\mright)$ with $X$ modal, called \emph{modal $A$-types}. It is a straightforward extension of the functor $\modal : \U \to \U_{\modal}$: for example, its object function sends $\mleft(X,s\mright)$ to $\mleft(\modal{X} , \eta_X \circ s \mright)$. By $\modal$-induction, we also have a family of equivalences $\mleft(\modal^A{U} \to_A V\mright) \xrightarrow{\simeq} \mleft(U \to_A V \mright)$ for $A$-types $U$ and modal $A$-types $V$ that is natural in $U$. (It's trivially natural in $V$.) Hence we have an adjunction between $\modal^A$ and the forgetful functor $\F_{\modal,A}$~\cite[\myhref{https://github.com/PHart3/colimits-agda/blob/lapc-lmcs/HoTT-Agda/theorems/modality/Mod-Cos-adj.agda\#L34}{Mod-cos-adj}]{2cohagda-pub}.
\begin{thm}[{\cite[\myhref{https://github.com/PHart3/colimits-agda/blob/lapc-lmcs/HoTT-Agda/theorems/modality/Mod-Cos-adj.agda\#L52}{Mod-cos-adj-2coh}]{2cohagda-pub}}] \label{mod2coh}
The left adjoint $\modal^A$ is $2$-coherent.
\end{thm}
\begin{proof}
By $\modal$-induction followed by a burst of path induction.
\end{proof}
\begin{cor}[{\cite[\myhref{https://github.com/PHart3/colimits-agda/blob/lapc-lmcs/HoTT-Agda/theorems/modality/Mod-colim.agda}{Mod-colim}]{2cohagda-pub}}]  \label{modcolbuild}
The wild category $\mleft(A/\U\mright)_{\modal}$ has all colimits.
\end{cor}
\begin{proof}
Let $\Gamma$ be a graph and $F$ a $\Gamma$-shaped diagram in $\mleft(A/\U\mright)_{\modal}$. By \cite[Theorem 15]{CSL25}, the diagram $\F_{\modal,A}(F)$ has a colimit $\colimm^A(\F(F))$ in $A/\U$. Since $\modal^A$ preserves colimits (\autoref{mod2coh}), we have the following colimiting cocone and natural isomorphism in $\mleft(A/\U\mright)_{\modal}$:
\[\begin{tikzcd}[column sep = 22, 
/tikz/column 1/.append style={column sep=11},
/tikz/column 2/.append style={column sep=11}]
	{\modal^A(\F(F_i))} && {\modal^A(\F(F_i))} & {F_i} && {F_j} \\
	& {\modal^A(\colimm^A(\F(F)))} && {\modal^A(\F(F_i))} && {\modal^A(\F(F_j))}
	\arrow["{{{\modal^A(\F(F_{i,j,g}))}}}", from=1-1, to=1-3]
	\arrow[from=1-1, to=2-2]
	\arrow[from=1-3, to=2-2]
	\arrow["{{{F_{i,j,g}}}}", from=1-4, to=1-6]
	\arrow["\simeq"', from=1-4, to=2-4]
	\arrow["{\eta^A}", from=1-4, to=2-4]
	\arrow["\simeq", from=1-6, to=2-6]
	\arrow["{\eta^A}"', from=1-6, to=2-6]
	\arrow["{{{\modal^A(\F(F_{i,j,g}))}}}"', from=2-4, to=2-6]
\end{tikzcd}\]
where $\eta^A(Z) \coloneqq \mleft(\eta_{\ty(Z)} , \refl\mright) : Z \to_A \modal^A(Z)$ for all $Z : A/\U$. 
It suffices to prove composing with natural isomorphisms preserves colimiting cocones.
This property has a simple proof in univalent wild bicategories, such as $\mleft(A/\U\mright)_{\modal}$ (given the univalence axiom). Indeed, univalence reduces the natural isomorphism $F \Rightarrow \modal^A \circ \F$ to the identity, and a standard bicategorical property implies that composing with the identity preserves colimiting cocones.
\end{proof}

Our construction of colimits of modal types is simpler than the Book proof (see \autoref{modsum}) and illuminates a higher coherence used by the latter. Equality (7.4.11) of the Book proof secretly requires a coherence condition between $\lN{-}\rN_n$'s composition law and its naturality data $\nat_n$, which is satisfied because $\lN{g \circ f}\rN_n \circ \lvert{-}\rvert_n \equiv \lvert{-}\rvert_n \circ g \circ f$. The required coherence has a similar flavor to \autoref{2coherdef}, but our proof makes such a condition explicit.

\section{Conclusion and future work}

We addressed a coherence problem in the proof that a left adjoint between wild categories preserves colimits. We proved that the coherence, which always holds in the classical setting, may be false for wild categories. We even supplied a wild left adjoint that outright fails to preserve colimits. With just ``off-the-shelf'' tools from HoTT, we identified a relatively tractable sufficient condition on the left adjoint for the proof of \emph{LAPC} to work, namely \emph{$2$-coherence}. We showed that the suspension functor is $2$-coherent and thus preserves colimits. In doing so, we managed to avoid an infeasible equality proof by developing a higher-dimensional version of Cavallo's trick for homogeneous types. We verified that a similar argument works for the join functor. Finally, we showed that modalities on coslices of a universe are $2$-coherent and, as a result, that the associated subcategories of modal types are cocomplete.

There are a few open questions raised by our work. The simplest is the analysis of the dual statement that right adjoints preserve limits for wild categories, which should be similar to the one presented here. Another question is whether we can extend \autoref{mod2coh} to all reflective subuniverses.
Finally, it would be quite useful to find a trick to show, in Book HoTT, that the smash product $X \land {-} : \U_{\ast} \to \U_{\ast}$ is $2$-coherent for each pointed type $X$. Like the join, the smash product generalizes the suspension~\cite[Proposition 4.2.1]{brunthesis}. The right adjoint of the smash product, however, is the pointed map space $X \to_{\ast} {-}$~\cite[Theorem 4.3.28]{vanDoornthesis}, which is not generally valued in homogeneous types. Hence we cannot use \autoref{trick} to escape the infeasible equality proof, which is likely even harder than the ones for the suspension and join.

\section*{Acknowledgment}
This material is based upon work supported by the Air Force Office of Scientific Research under award number FA9550-21-1-0009.
Any opinions, findings, and conclusions or recommendations expressed in this material are those of the author(s) and do not
necessarily reflect the views of the United States Air Force.

\bibliographystyle{alphaurl}
\bibliography{lapc}

\end{document}